\documentclass{article}

\usepackage{PRIMEarxiv}

\usepackage[utf8]{inputenc} % allow utf-8 input
\usepackage[T1]{fontenc}    % use 8-bit T1 fonts
\usepackage{hyperref}       % hyperlinks
\usepackage{url}      
\usepackage{bbm}% simple URL typesetting
\usepackage{booktabs}       % professional-quality tables
\usepackage{amsfonts}       % blackboard math symbols
\usepackage{nicefrac} 
\usepackage{amsthm}
\usepackage{amsmath}% compact symbols for 1/2, etc.
\usepackage{microtype}      % microtypography
\usepackage{lipsum}
\usepackage{fancyhdr}       % header
\usepackage{graphicx}   
\usepackage{subcaption}% graphics
\graphicspath{{media/}}
\newtheorem{theorem}{Theorem}[section]
\newtheorem{corollary}{Corollary}[section]
\newtheorem{definition}{Definition}[section]
\newtheorem{remark}{Remark}[section]
\title{Persistence and Extinction Dynamics in a Stochastic Predator-Prey Model with Emergent Allee Effects
}

\author{
  Carlos Granados\\
  Institute of Mathematics\\
  University of Antioquia \\
  Medellin-Colombia\\
  \texttt{carlosgranadosortiz@outlook.es} \\
   \And
  Leon A. Valencia \\
  Institute of Mathematics\\
  University of Antioquia \\
  Medellin-Colombia\\
  \texttt{lalexander.valencia@udea.edu.co} \\
}

\begin{document}
\maketitle

\begin{abstract}
The Allee effect describes a decline in population fitness at low densities, potentially leading to extinction. In predator-prey systems, an emergent Allee effect can arise due to interactions such as density-dependent maturation rates and predation constraints. This work studies a stochastic predator-prey model where the prey population is structured into juvenile and adult stages, with maturation following a nonlinear function. We introduce Itô-type stochastic perturbations in mortality rates to account for environmental variability. We first establish the positivity of solutions and derive sufficient conditions for the stability of the trivial equilibrium, prey extinction, and conditional predator extinction. We then analyze prey persistence under specific maturation rate functions. Finally, numerical simulations illustrate the theoretical results and their ecological implications.
\end{abstract}

% keywords can be removed

\section{Introduction}

The Allee effect is a fundamental concept in population dynamics, describing a phenomenon where individuals of a species experience reduced biological fitness at low population densities, leading to a critical threshold below which the population faces an increased risk of extinction \cite{allee1931, courchamp2008}. This effect can arise from various ecological mechanisms, including difficulties in finding mates, reduced cooperation in foraging or defense, and inbreeding depression \cite{stephens1999}. 

In his seminal work of 1931, W.C. Allee provided a conceptual and empirical foundation for this phenomenon, observing that animal populations become more vulnerable to extinction at low densities. However, it was not until later that the Allee effect was formalized mathematically using ordinary differential equations. One of the most widely recognized models capturing this effect takes the form:

\[
\frac{dN}{dt} = rN\left(1 - \frac{N}{K}\right) - \frac{aN^2}{N + b},
\]

where \( N \) represents the population density, \( r \) is the intrinsic growth rate, \( K \) is the carrying capacity, \( a \) governs the intensity of the Allee effect, and \( b \) modulates its influence. The term \( \frac{aN^2}{N + b} \) introduces the Allee effect by reducing the population growth rate at low densities, reflecting challenges such as mate finding or resource acquisition that can drive populations toward decline or extinction.

An emergent Allee effect occurs when interactions between species produce population dynamics that exhibit Allee-like characteristics under specific conditions. For instance, in predator-prey systems, such an effect can arise when:

\begin{itemize}
    \item The prey's reproduction rate depends positively on population density, as higher densities may enhance defense mechanisms, foraging efficiency, or communal behavior.
    \item Predators exhibit higher hunting efficiency when prey are clustered or exhibit communal behaviors that facilitate capture.
\end{itemize}

These dynamics can lead to pronounced cycles of abundance and scarcity in both predator and prey populations, potentially resulting in the extinction of one or both species. 

Of particular interest is the emergent Allee effect in predator-prey models, which has been extensively studied using ordinary and stochastic differential equations \cite{Kong2024, Xue2024, sen2021bifurcation, mulugeta2022bifurcation, stochastic_allee_extinction, stochastic_allee_dynamics, stochastic_allee_levy, stochastic_allee_holling, stochastic_allee_herd, Pimentel2020, vanKooten2005,bistability_allee_stage_predation}. This work is inspired by the deterministic model proposed in \cite{vanKooten2005}, which divides the prey's life history into juvenile and adult stages. The prey population is regulated through maturation, with predators feeding exclusively on adult prey. This framework captures the essential elements for the emergence of an Allee effect, as detailed in \cite{DeRoos2002}. The model is described by the following system of equations:

\begin{align*}
    \frac{dx}{dt} &= \beta y - \frac{x}{1 + x^2} - \mu_x x, \\
    \frac{dy}{dt} &= \frac{x}{1 + x^2} - \mu_y y - y z, \\
    \frac{dz}{dt} &= z(\alpha y - \mu_z),
\end{align*}

where \( z \) represents the predator population, \( y \) denotes adult prey, and \( x \) corresponds to juvenile prey. The maturation rate of juveniles is given by \( \frac{1}{1 + x^2} \), which decreases as the density of the juvenile population increases. The parameter \( \beta \) represents the fertility rate, while \( \alpha \) quantifies the efficiency with which the consumed prey biomass is converted to predator biomass ($\alpha\in(0,1)$). The mortality rates for juveniles, adult prey and predators are denoted by \( \mu_x \), \( \mu_y \), and \( \mu_z \), respectively.

In this study, we generalize the maturation rate to a function \( \varphi: \mathbb{R}_{\geq 0} \rightarrow \mathbb{R}_{\geq 0} \) that satisfies the following conditions:

\begin{enumerate}
    \item \( \varphi \) is differentiable on \( (0, \infty) \),
    \item \( \lim_{x \to \infty} \varphi(x) = 0 \), and
    \item \( \sup_{x \in \mathbb{R}_{\geq 0}} \varphi(x) \leq 1 \).
\end{enumerate}

To better reflect the stochastic nature of ecological systems, we introduce stochastic perturbations of Itô-type into mortality rates \( \mu_x \), \( \mu_y \), and \( \mu_z \). These perturbations account for environmental variability, demographic noise, and other stochastic factors, modeled as follows:

\begin{align*}
    \mu_x &\mapsto \mu_x + \sigma_x dB_x(t), \\
    \mu_y &\mapsto \mu_y + \sigma_y dB_y(t), \\
    \mu_z &\mapsto \mu_z + \sigma_z dB_z(t),
\end{align*}

where \( B_x \), \( B_y \), and \( B_z \) are independent Brownian motions, and \( \sigma_x \), \( \sigma_y \), and \( \sigma_z \) represent the noise intensity coefficients. The resulting stochastic model is given by:

\begin{align}\label{themodel}
    dx &= \left(\beta y - \varphi(x)x - \mu_x x\right) dt - \sigma_x x dB_x(t), \nonumber \\
    dy &= \left(\varphi(x)x - \mu_y y - y z\right) dt - \sigma_y y dB_y(t), \\
    dz &= z(\alpha y - \mu_z) dt - \sigma_z z dB_z(t). \nonumber
\end{align}

We begin by establishing the positivity of the solutions, demonstrating that the system \eqref{themodel} remains within the domain of positive population densities for all time with probability 1. Subsequently, we derive sufficient conditions for the stability of the origin, the asymptotic stability of the origin, the extinction of the prey, and the conditional extinction of the predator.

A central focus of this work is on the analysis of prey persistence. To this end, we consider the specific case \( \varphi(x) = \frac{\kappa}{1 + x} \), where \( \kappa \) represents the maturation rate at low densities. This choice is motivated by observations in \cite{bistability_allee_stage_predation}, where the population density of cod in the North Atlantic remained low despite fishing restrictions. We model predation constraints by imposing the condition \( \{\sup_{t \geq 0} z(t) \leq M\} \), aiming to identify conditions under which prey persistence is ensured with limited predation.

The remainder of this article is organized as follows. 

\begin{itemize}
    \item In Section \ref{sec:preliminares}, we present the necessary definitions and theoretical foundations for the study.
    \item In Section \ref{sec:main results}, we establish sufficient conditions for the stability of the origin, the extinction of the prey, and the conditional extinction of the predator.
    \item In Section \ref{sec:persistencia}, we analyze prey persistence and derive conditions under which extinction is avoided.
    \item In Section \ref{sec:simulaciones}, we provide numerical simulations to illustrate the theoretical results, make comments on the theorems, and discuss their biological implications.   
    
\end{itemize}

\section{Preliminaries}\label{sec:preliminares}

Throughout this paper, let $\mathbb{R}^n_{\geq 0}:= \{  \mathbf{x}\in \mathbb{R}^n: x_i\geq 0, \forall i=1,\ldots,n\}$ and $\mathbb{R}^n_{> 0}:= \{  \mathbf{x}\in \mathbb{R}^n: x_i> 0, \forall i=1,\ldots,n\}$.

Consider the following Itô-type stochastic differential equation:

\begin{equation}\label{ede}
    dX(t) = f(X(t), t) \, dt + g(X(t), t) \, dB(t),
\end{equation}

where \( f: \mathbb{R}^n \times \mathbb{R}_{\geq 0} \rightarrow \mathbb{R}^n \) and \( g: \mathbb{R}^n \times \mathbb{R}_{\geq 0} \rightarrow M_{n \times m} \) are Borel-measurable functions, with \( M_{n \times m} \) being the space of \( n \times m \) matrices with real entries, and \( \{B(t)\} \) an \( m \)-dimensional Brownian motion.

The infinitesimal generator operator \( \mathcal{L} \) associated with \eqref{ede}, whose domain consists of smooth functions, takes the form:

\begin{definition}\label{def1}
    A point $\mathbf{x}=\mathbf{0}\in\mathbb{R}^n$ is said to be an equilibrium point of the system \eqref{ede} if $f(\mathbf{0},t)=\mathbf{0}$ and $g(\mathbf{0},t)=\mathbf{0}$ for all $t\in \mathbb{R}_{\geq 0}$. 
\end{definition}

\begin{definition}\label{def2}
    Let $\mathbf{x}=\mathbf{0}$ be an equilibrium point of the system \eqref{ede}. This point is said to be stable in probability if for every $\epsilon>0$ and $r>0$, there exists $\delta=\delta(\epsilon,r)>0$ such that if $\|\mathbf{x}(0)\|<\delta$, then 
    
    $$\mathbb{P}(\|\mathbf{x}(t)\|<r, \forall t\in \mathbb{R}_{\geq 0})\geq 1-\epsilon.$$
\end{definition}

Since the stochastic differential equation of our interest models a population growth phenomenon, the point $\mathbf{x}=\mathbf{0}$ represents the absence of species (extinction). Thus, Definition \eqref{def2} intuitively tells us that the system will remain within thresholds close to extinction with high probability, provided it starts with low population densities.

\begin{definition}\label{def3}
    Let $\mathbf{x}=\mathbf{0}$ be an equilibrium point of the system \eqref{ede}. We say that this point is asymptotically stable in probability if it is stable in probability and if for every $\epsilon>0$, there exists $\delta=\delta(\epsilon)>0$ such that if $\|\mathbf{x}(0)\|<\delta$, then
    $$ \mathbb{P}\left(\displaystyle \lim_{t\to\infty} \mathbf{x}(t)=\mathbf{0}\right)\geq 1-\epsilon. $$
\end{definition}

The definition \eqref{def3} intuitively tells us that being asymptotically stable in probability implies that, with high probability, there exists a neighborhood of the origin ($\mathbf{x}=\mathbf{0}$ is an equilibrium point of the system) such that, starting from this neighborhood, the system converges toward extinction.

 Next, we will state a series of well-known theorems in the literature on stochastic differential equations that will be used in the proofs in the main results section. However, we first need some additional definitions.

 \begin{definition}
    Let \(\mathcal{K}\) be the family of all continuous, non-decreasing functions \(\phi: \mathbb{R}_{\geq 0} \rightarrow \mathbb{R}_{\geq 0}\) such that \(\phi(0) = 0\) and \(\phi(r) > 0\) if \(r > 0\). For \(h > 0\), let \(S_h = \{ \mathbf{x} \in \mathbb{R}^n: \|\mathbf{x}\| < h \}\). A continuous function \(V(\mathbf{x}, t)\) defined on \(S_h \times \mathbb{R}_{\geq 0}\) is said to be positive definite (in the sense of Lyapunov) if \(V(\mathbf{0}, t) \equiv 0\) and, for some \(\phi \in \mathcal{K}\),
    \[
    V(\mathbf{x}, t) \geq \phi(\|\mathbf{x}\|) \quad \text{for all } (\mathbf{x}, t) \in S_h \times \mathbb{R}_{\geq 0}.
    \]
    A continuous function \(V(\mathbf{x}, t)\) is said to be negative definite if \(-V\) is positive definite. A nonnegative continuous function \(V(\mathbf{x}, t)\) is said to be decrescent if for some \(\phi \in \mathcal{K}\),
    \[
    V(\mathbf{x}, t) \leq \phi(\|\mathbf{x}\|) \quad \text{for all } (\mathbf{x}, t) \in S_h \times \mathbb{R}_{\geq 0}.
    \]
\end{definition}

 \begin{definition}
    Let \(h \in \mathbb{R}_{>0}\). We denote by \(C^{2,1}(S_h \times \mathbb{R}_{\geq 0}; \mathbb{R}_{\geq 0})\) the space of all nonnegative functions \(V(\mathbf{x}, t)\) defined on \(S_h \times \mathbb{R}_{\geq 0}\) that are twice continuously differentiable in \(\mathbf{x}\) and once continuously differentiable in \(t\).
\end{definition}

\begin{theorem}\label{thm1}
    Let \(\mathbf{x} = \mathbf{0}\) be an equilibrium point of \eqref{ede}. If there exists a positive definite function \(V(\mathbf{x}, t) \in C^{2,1}(S_h \times \mathbb{R}_{\geq 0}; \mathbb{R}_{\geq 0})\) such that
    \[
    \mathcal{L}V(\mathbf{x}, t) \leq 0
    \]
    for all \((\mathbf{x}, t) \in S_h \times \mathbb{R}_{\geq 0}\), then \(\mathbf{x} = \mathbf{0}\) is stable in probability.
\end{theorem}
\begin{proof}
    See \cite{mao2007stochastic}, page 111.
\end{proof}

Consider the linear system

\begin{equation}\label{lineal}
    dX(t) = BX \, dt + \sum_{r=1}^k \mathbf{c}_r X \, dB_r(t)
\end{equation}
with constant coefficients, that is, \( B, \mathbf{c}_1, \ldots, \mathbf{c}_k \) are constant matrices.

\begin{theorem}\label{thm2}
    If \(\mathbf{x} = \mathbf{0}\) is an equilibrium point of the system \eqref{lineal} that is asymptotically stable in probability, and the coefficients of the system

\begin{equation}\label{ede1}
    dX(t) = f(X(t), t) \, dt + \sum_{r=1}^k g_r(X(t), t) \, dB_r(t)
\end{equation}
satisfy the inequality

\begin{equation}\label{ine1}
    \|f(\mathbf{x}, t) - B\mathbf{x}\| + \sum_{r=1}^k \|g_r(\mathbf{x}, t) - \mathbf{c}_r \mathbf{x}\| < \gamma \|\mathbf{x}\|
\end{equation}
in a sufficiently small neighborhood of the point \(\mathbf{x} = \mathbf{0}\) and for a sufficiently small constant \(\gamma\), then the equilibrium point \(\mathbf{x} = \mathbf{0}\) is asymptotically stable in probability for the system \eqref{ede1}.
\end{theorem}

\begin{proof}
    See \cite{khasminskii2012stochastic}, Theorem 7.1, p. 228.
\end{proof}

\begin{theorem}\label{thm3}
    If there exists a positive definite and decrescent function \( V(\mathbf{x}, t) \in C^{2,1}(\mathbb{R}^n \times \mathbb{R}_{\geq 0}, \mathbb{R}_{\geq 0}) \) such that \( \mathcal{L}V(\mathbf{x}, t) \) is negative definite, then the equilibrium point \( \mathbf{x} = \mathbf{0} \) is asymptotically stable in probability.
\end{theorem}
\begin{proof}
    See \cite{mao2007stochastic}, Theorem 2.3, p. 112.
\end{proof}

\section{Main results}\label{sec:main results}

The techniques to be used in the majority of the proofs rely on applying the generator of the stochastic process associated with \eqref{themodel} (the Lyapunov operator), which has the following structure:

\begin{eqnarray}\label{lyapunovoperator}
    \mathcal{L}&=\frac{\partial}{\partial t}+\left( \beta y -\varphi(x)x- \mu_x x \right)\frac{\partial}{\partial x}+\left(\varphi(x)x-\mu_y y-yz\right)\frac{\partial}{\partial y}+(\alpha yz-\mu_z z)\frac{\partial}{\partial z}\\
    &+\frac{\sigma_x^2}{2}x^2\frac{\partial^2}{\partial x^2}+\frac{\sigma_y^2}{2}y^2\frac{\partial^2}{\partial y^2}+\frac{\sigma_z^2}{2}z^2\frac{\partial^2}{\partial z^2}\nonumber
\end{eqnarray}

whose domain is defined over the space \( C^{2,1}(x, y, z, t) \), that is, the space of functions that are twice continuously differentiable in \( (x, y, z) \in \mathbb{R}^3 \) and once continuously differentiable in \( t \in \mathbb{R} \).

The following theorem guarantees that the solution of the stochastic differential equation does not explode in finite time and remains in the positive cone \(\mathbb{R}^3_{>0} \).

\begin{theorem}[Positivity of the solution]\label{positividad}
    For initial values \((x_0, y_0, z_0) \in \mathbb{R}^3_{> 0}\), the system \eqref{themodel} has a unique solution \((x(t), y(t), z(t))\) for all \(t \in \mathbb{R}_{\geq 0}\), and the solution remains in \(\mathbb{R}^3_{> 0}\) with probability \(1\), that is,
    \[
    \mathbb{P}\left((x(t), y(t), z(t)) \in \mathbb{R}^3_{> 0}, \forall t \in \mathbb{R}_{> 0}\right) = 1.
    \]
\end{theorem}

\begin{proof}
The proof is classical in the theory of Itô-type stochastic differential equations.  
Consider the function \newline \( V: \mathbb{R}^3_{>0} \rightarrow \mathbb{R}_{\geq 0} \) defined by:

\[
V(x, y, z) = x - 1 - \ln(x) + y - 1 - \ln(y) + z - 1 - \ln(z).
\]

Then,

\[
\begin{aligned}
\mathcal{L}V(x, y, z) &= \beta y - \varphi(x)x - \mu_x x - \frac{\beta y}{x} + \varphi(x) + \mu_x + \varphi(x)x - \mu_y y - yz \\
&\quad - \frac{\varphi(x)x}{y} + \mu_y + z + \alpha z y - \mu_z z - y + \mu_z + \frac{\sigma_x^2}{2} + \frac{\sigma_y^2}{2} + \frac{\sigma_z^2}{2}.
\end{aligned}
\]

Therefore, for \( (x, y, z) \in \mathbb{R}^3_{>0} \), using that \( \alpha \in (0, 1) \) and that \( \varphi(x) \leq 1 \), we obtain:

\begin{eqnarray*}
    \mathcal{L}V(x,y,z)&\leq& \beta y +z +\frac{1}{2}(\sigma_x^2+\sigma_y^2+\sigma_z^2)+\mu_x+\mu_y+\mu_z+1\\
    & \leq & \max\{1,\beta\}(y+z)+\frac{1}{2}(\sigma_x^2+\sigma_y^2+\sigma_z^2)+\mu_x+\mu_y+\mu_z+1\\
    & \leq & \max\{1,\beta\}(x+y+z)+\frac{1}{2}(\sigma_x^2+\sigma_y^2\sigma_z^2)+\mu_x+\mu_y+\mu_z+1\\
    & \leq & 2\max\{1,\beta\}(V(x,y,z)+3)+\frac{1}{2}(\sigma_x^2+\sigma_y^2+\sigma_z^2)+\mu_x+\mu_y+\mu_z+1\\
    &=& c_1V(x,y,z)+c_2.
\end{eqnarray*}

Where \( c_1 = 2\max\{1, \beta\} \) and \( c_2 = 6\max\{1, \beta\} + \frac{1}{2}(\sigma_x^2 + \sigma_y^2 + \sigma_z^2) + \mu_x + \mu_y + \mu_z + 1 \).

Now, let \( (x_0, y_0, z_0) \in \mathbb{R}^3_+ \) and \( n_0 \in \mathbb{N} \) be sufficiently large such that \( \min\{x_0, y_0, z_0\} > 1/n \) for all \( n \geq n_0 \). Since the coefficients of \eqref{themodel} are locally Lipschitz functions, the solution of \eqref{themodel} starting from \( (x_0, y_0, z_0) \) exists and is unique, up to a stochastic time \( \tau_e \), called the explosion time. Additionally, we will prove that with probability \( 1 \), the solution remains in \( \mathbb{R}^3_{>0} \) with initial condition \( (x_0, y_0, z_0) \in \mathbb{R}^3_{>0} \). We define the following stopping time for all \( n \geq n_0 \):

\[
\tau_n := \inf\{t \in \mathbb{R}_{\geq 0} : (x(t), y(t), z(t)) \notin [1/n, n]^3\}.
\]

Note that \( \tau_n \) is increasing in \( n \), and therefore let \( \tau_\infty := \lim_{n \to \infty} \tau_n \). It is clear that
\[
\tau_\infty \leq \tau_e := \inf\{t \in \mathbb{R}_{\geq 0} : \|(x(t), y(t), z(t))\| = +\infty\}.
\]
It is then sufficient to show that \( \mathbb{P}(\tau_\infty = \infty) = 1 \). 

To proceed, assume by contradiction that \( \mathbb{P}(\tau_\infty = \infty) < 1 \). Then, there exist \( T > 0 \) and \( \eta > 0 \) such that \( \mathbb{P}(\tau_\infty < T) > \eta \). Consequently, there exists an \( n_1 \geq n_0 \) such that
\[
\mathbb{P}(\tau_n < T) > \eta, \quad \text{for all } n \geq n_1.
\]

By applying Itô's formula, one obtains:
\[
dV(x(t), y(t), z(t)) = \mathcal{L}V(x(t), y(t), z(t)) \, dt + \frac{\partial V}{\partial x} \, dB_x(t) + \frac{\partial V}{\partial y} \, dB_y(t) + \frac{\partial V}{\partial z} \, dB_z(t).
\]
It follows that \( V(x(\tau_n \wedge T), y(\tau_n \wedge T), z(\tau_n \wedge T)) \) is equal to
\[
\begin{aligned}
&V(x(0), y(0), z(0)) + \int_0^{\tau_n \wedge T} \mathcal{L}V(x(\tau_n \wedge t), y(\tau_n \wedge t), z(\tau_n \wedge t)) \, dt \\
&+ \int_0^{\tau_n \wedge T} \frac{\partial V}{\partial x} \, dB_x(t) 
+ \int_0^{\tau_n \wedge T} \frac{\partial V}{\partial y} \, dB_y(t) 
+ \int_0^{\tau_n \wedge T} \frac{\partial V}{\partial z} \, dB_z(t).
\end{aligned}
\]

By taking the expected value, and using the fact that \( \mathcal{L}V(x, y, z) \leq c_1 V(x, y, z) + c_2 \), \( V \) is a non-negative function, \( c_1, c_2 \) are positive, and the Itô integral has zero mean, one obtains:
\begin{eqnarray*}   
        \mathbb{E}V(x(\tau_n\wedge T),y(\tau_n\wedge T), (z(\tau_n\wedge T))&\leq & V(x(0),y(0),z(0))+\mathbb{E}\left[ \int_0^{\tau_n\wedge T} (c_1 V(x(t),y(t),z(t))+c_2)dt\right]\\
        & \leq &  V(x(0),y(0),z(0))+ \mathbb{E}\left[ \int_0^{ T} (c_1 V(x(t\wedge \tau_n),y(t\wedge \tau_n),z(t\wedge \tau_n))+c_2)dt\right]\\
        & \leq &  V(x(0),y(0),z(0))+c_2T+c_1\mathbb{E}\left[ \int_0^{ T} V(x(t\wedge \tau_n),y(t\wedge \tau_n),z(t\wedge \tau_n))dt\right]\\
        & \leq &  V(x(0),y(0),z(0))+c_2T+c_1\int_0^T\mathbb{E}V(x(t\wedge \tau_n),y(t\wedge \tau_n),z(t\wedge \tau_n))dt.
\end{eqnarray*}
The last line follows from Fubini's theorem. Then, using Gronwall's inequality, one obtains:

\begin{equation}\label{eq1}
     \mathbb{E}V(x(\tau_n \wedge T), y(\tau_n \wedge T), z(\tau_n \wedge T)) \leq \left[ V(x(0), y(0), z(0)) + c_2 T \right] e^{c_1 T}.
\end{equation}

Let \( h(n) = \min\{1/n - 1 - \ln(1/n), n - 1 - \ln(n)\} \). Using the equality above and the fact that \( V \) is a non-negative function, it follows that:

\begin{equation}\label{eq2}
     \mathbb{E}\left[ V(x(\tau_n \wedge T), y(\tau_n \wedge T), z(\tau_n \wedge T)) \mathbbm{1}_{\{\tau_n \leq T\}} \right] \geq h(n) \mathbb{P}(\tau_n < T) \geq h(n) \eta.
\end{equation}

Inequality \eqref{eq2} implies that the sequence \( \left\{ \mathbb{E}\left[ V(x(\tau_n \wedge T), y(\tau_n \wedge T), z(\tau_n \wedge T)) \mathbbm{1}_{\{\tau_n \leq T\}} \right] \right\}_{n \geq n_1} \) is unbounded above. However, \eqref{eq1} contradicts this. Therefore, \( \mathbb{P}(\tau_\infty = \infty) = 1 \), or equivalently,

\[
\mathbb{P}\left( (x(t), y(t), z(t)) \in \mathbb{R}^3_{> 0}, \forall t \in \mathbb{R}_{\geq 0} \right) = 1.
\]
\end{proof}

\begin{remark}\label{obs1}
   Due to the previous theorem, that is, for every initial condition \( (x_0, y_0, z_0) \in \mathbb{R}^3_{> 0} \), the unique solution of the system \eqref{themodel} remains in \( \mathbb{R}^3_{> 0} \) for all time \( t \geq 0 \), the domain of definition of the function \( V \) in Theorems \eqref{thm1}, \eqref{thm2}, and \eqref{thm3} can be changed from \( \mathbb{R}^3 \) to \( \mathbb{R}^3_{\geq 0} \). To understand why, one must read the proofs of the theorems. However, informally speaking, the proofs use the composition of the function \( V \) with the solution \( (x(t), y(t), z(t)) \), and in this composition the values of \( V \) outside \( \mathbb{R}^3_{\geq 0} \) are irrelevant. Therefore, \( S_h \) can be replaced by \( S_h \cap \mathbb{R}^3_{\geq 0} \).
\end{remark}

\vspace{0.5cm}

\begin{theorem}\label{SP}[Stability in probability of $\mathbf{x} = (0, 0, 0)$.]
   Suppose that in \eqref{ede} condition \( \mu_y \geq \beta \) holds. Then, the equilibrium point \( \mathbf{x} = \mathbf{0} \) is stable in probability.
\end{theorem}

\begin{proof}
    Let \( V(x, y, z, t) = x + y + z \in C^2(\mathbb{R}^3_{\geq 0} \times \mathbb{R}_{\geq 0}, \mathbb{R}_{\geq 0}) \). Note that \( V(0, 0, 0, t) \equiv 0 \). Consider the function \( \phi(r) = \frac{1}{\sqrt{3}} r^{1/2} \) for \( r \geq 0 \). Then, \( \phi \in \mathcal{K} \) and \( V(x, y, z, t) \geq \phi(\|(x, y, z)\|) \) for \( (x, y, z, t) \in \mathbb{R}^3_{> 0} \times \mathbb{R}_{> 0} \). Due to Remark \eqref{obs1}, Theorem \eqref{thm1} can be applied.

The generator of the process \eqref{themodel} is given by \eqref{lyapunovoperator}. Thus,

\[
\mathcal{L}V(x, y, z, t) = \beta y - \mu_x x - \mu_y y - y z - \alpha y z - \mu_z.
\]

Therefore, if \( \beta \leq \mu_y \) and recalling that \( \alpha \in (0, 1) \), it follows that \( \mathcal{L}V \leq 0 \) on \( \mathbb{R}^3_{\geq 0} \times \mathbb{R}_{\geq 0} \). The result follows from Theorem \eqref{thm1}.
\end{proof}

\begin{remark}
    The condition $\mu_y \geq \beta$ in Theorem \ref{SP} is quite reasonable. If the constant $\beta$, which represents the fertility rate, does not exceed the mortality rate of adult prey, then it is expected that if the population density of both juvenile and adult prey starts at low thresholds, it will remain at low thresholds with high probability. Furthermore, in the absence of a sufficient number of adult prey, the population density of predators is also expected to remain at levels close to extinction.  

\end{remark}

\begin{theorem}\label{sap}[Asymptotic stability in probability of $\mathbf{x} = (0, 0, 0)$.]
    Consider the system \eqref{themodel} such that:
    \begin{enumerate}
        \item \( \mu_x > \frac{\sigma_x^2}{2} \) and \( \mu_z > \frac{\sigma_z^2}{2} \), and
        \item \( \left( \frac{\sigma_x^2}{2} - \mu_x \right) \left( \frac{\sigma_y^2}{2} - \mu_y \right) > \frac{\beta^2}{4} \).
    \end{enumerate}
    Then, the equilibrium point of the system \eqref{themodel}, \( \mathbf{x} = (0, 0, 0) \), is asymptotically stable in probability.
\end{theorem}

\begin{proof}
   Consider the following stochastic differential equation:

\begin{align}\label{themodel1}
  dx &= \left(\beta y - \mu_x x\right) dt - \sigma_x x \, dB_x(t), \nonumber \\
  dy &= -\mu_y y \, dt - \sigma_y y \, dB_y(t), \\
  dz &= \mu_z z \, dt - \sigma_z z \, dB_z(t). \nonumber
\end{align}

It will be shown that \( \mathbf{x} = \mathbf{0} \) is asymptotically stable in probability for \eqref{themodel1}. Let

\[
\hat{\mathcal{L}} = \frac{\partial}{\partial t} + \left(\beta y - \mu_x x\right) \frac{\partial}{\partial x} - \mu_y y \frac{\partial}{\partial y} - \mu_z z \frac{\partial}{\partial z} + \frac{\sigma_x^2}{2} \frac{\partial^2}{\partial x^2} + \frac{\sigma_y^2}{2} \frac{\partial^2}{\partial y^2} + \frac{\sigma_z^2}{2} \frac{\partial^2}{\partial z^2}.
\]

Let \( h > 0 \) be arbitrary, and consider the positive-definite function \( V(x, y, z, t) = \frac{1}{2}(x^2 + y^2 + z^2) \) for \( (x, y, z, t) \in S_h \cap \mathbb{R}^3_{\geq 0} \times \mathbb{R}_{\geq 0} \). Now,

\[
\hat{\mathcal{L}}V(x, y, z, t) = -\left(\mu_x - \frac{\sigma_x^2}{2}\right) x^2 - \left(\mu_y - \frac{\sigma_y^2}{2}\right) y^2 - \left(\mu_z - \frac{\sigma_z^2}{2}\right) z^2 + \beta x y.
\]

Thus, \( \hat{\mathcal{L}}V(x, y, z, t) \) can be represented in the form

\[
\hat{\mathcal{L}}V(x, y, z, t) = A x^2 + B x y + C y^2 + D z^2,
\]

where

\begin{itemize}
    \item \( A = -\mu_x + \frac{1}{2} \sigma_x^2 \),
    \item \( B = \beta \),
    \item \( C = -\mu_y + \frac{1}{2} \sigma_y^2 \), and
    \item \( D = -\mu_z + \frac{1}{2} \sigma_z^2 \).
\end{itemize}

Consider the matrix \( Q \) given by:

\[
Q = \left(
\begin{array}{ccc}
   2A  & B & 0  \\
   B & 2C & 0 \\
   0 & 0 & 2D
\end{array}
\right).
\]

Then, \( \hat{\mathcal{L}}V(x, y, z) \) can be represented as the quadratic form \( \frac{1}{2} (x, y, z) Q (x, y, z)^{tr} \). Due to the inequality for symmetric matrices \( M \), \( \lambda_{\min} \|\mathbf{x}\|^2 \leq \mathbf{x} M \mathbf{x}^{tr} \leq \lambda_{\max} \|\mathbf{x}\|^2 \), where \( \lambda_{\min} \) and \( \lambda_{\max} \) are the smallest and largest eigenvalues of \( M \), respectively, it follows that the function \( \hat{\mathcal{L}}V(x, y, z) \) is negative definite in the Lyapunov sense if the matrix \( Q \) is negative definite (all its eigenvalues are negative).

The principal minors of the matrix \( Q \) are

\[
Q_1 = 2A, \quad Q_2 = \left( \begin{array}{cc}  
2A & B \\
B & 2C \\
\end{array}
\right), \quad Q_3 = Q.
\]

Now,

\begin{itemize}
    \item \( \det(Q_1) = 2A < 0 \) since \( \mu_x > \frac{\sigma_x^2}{2} \),
    \item \( \det(Q_2) = 4AC - B^2 > 0 \) since \( \left(\frac{\sigma_x^2}{2} - \mu_x\right) \left(\frac{\sigma_y^2}{2} - \mu_y\right) > \frac{\beta^2}{4} \), and
    \item \( \det(Q_3) = 2D \det(Q_2) < 0 \) since \( \det(Q_2) > 0 \) and \( 2D < 0 \) (because \( \mu_z > \frac{\sigma_z^2}{2} \)).
\end{itemize}

Since the determinants of the principal minors alternate in sign and \( \det(Q_1) < 0 \), the matrix \( Q \) is negative definite.

To use Theorem \ref{thm2}, it remains to prove that the inequality \eqref{ine1} holds for \eqref{themodel} and \eqref{themodel1}, which is equivalent to

\[
|\varphi(x) x| + |\varphi(x) x - y z| + |\alpha y z| \leq c (x^2 + y^2 + z^2)^{1/2}, \quad \forall (x, y, z) \in S_h \cap \mathbb{R}^3_{\geq 0},
\]

for some positive constants \( h, c \). Now, for all \( (x, y, z) \in S_1 \cap \mathbb{R}^3_{\geq 0} \), it holds that:

\[
\begin{aligned}
|\varphi(x) x| + |\varphi(x) x - y z| + |\alpha y z| &\leq 2 \varphi(x) x + (\alpha + 1) y z \\
&\leq 2x + 2y + 2z \\
&\leq 6 (x^2 + y^2 + z^2)^{1/2}.
\end{aligned}
\]

It is concluded, by Theorem \ref{thm2}, that \( \mathbf{x} = (0, 0, 0) \) is an asymptotically stable equilibrium point in probability.
\end{proof}

\begin{remark}
    In Theorem \ref{sap}, the interpretation of the sufficient conditions does not seem very natural. However, if $\sigma_x, \sigma_y$ is sufficiently small, it can be interpreted that if the product of mortality rates satisfies $\mu_x \mu_y > \beta^2 / 2$, then, starting from low population density thresholds, the system converges to extinction with high probability.  

\end{remark}

\begin{theorem}\label{ExtPrey}[Extinction of the prey]
   Let \( (x(t), y(t), z(t)) \) be the solution of \eqref{themodel} with initial condition \( (x_0, y_0, z_0) \in \mathbb{R}^3_{> 0} \). If \( \beta < \min\{\mu_x, \mu_y\} \), then the prey population will go extinct with probability \( 1 \), that is,
   \[
   \lim_{t \to \infty} [x(t) + y(t)] = 0 \quad \text{a.s.}
   \]
\end{theorem}

\begin{proof}
    Let \( V(x, y, z, t) = \ln(x + y) \). Note that, due to Theorem \eqref{positividad}, the composition of the function \( V \) with the solution of \eqref{themodel} is well defined with probability \( 1 \). Using Itô's formula, one obtains:

\[
\begin{aligned}
dV(x, y, z, t) &\leq \left[ \frac{\beta y - \mu_x x - \mu_y y - y z}{x + y} - \frac{1}{2} \frac{\sigma_x^2 x^2 + \sigma_y^2 y^2}{(x + y)^2} \right] dt \\
&\quad + \frac{\sigma_x x}{x + y} dB_x(t) + \frac{\sigma_y y}{x + y} dB_y(t).
\end{aligned}
\]

Thus,

\[
\begin{aligned}
dV(x, y, z, t) &\leq \left[ \frac{\beta y - \mu_x x - \mu_y y}{x + y} \right] dt + \frac{\sigma_x x}{x + y} dB_x(t) + \frac{\sigma_y y}{x + y} dB_y(t) \\
&\leq [\beta - (\mu_x - \mu_y)] dt + \frac{\sigma_x x}{x + y} dB_x(t) + \frac{\sigma_y y}{x + y} dB_y(t) \\
&\leq (\beta - \min\{\mu_x, \mu_y\}) dt + \frac{\sigma_x x}{x + y} dB_x(t) + \frac{\sigma_y y}{x + y} dB_y(t).
\end{aligned}
\]

Therefore,

\begin{align}\label{eq3}
\ln[x(t) + y(t)] &\leq \ln [x_0 + y_0] + (\beta - \min\{\mu_x, \mu_y\}) t \nonumber \\
&\quad + \int_0^t \frac{\sigma_x x(s)}{x(s) + y(s)} dB_x(s) + \int_0^t \frac{\sigma_y y(s)}{x(s) + y(s)} dB_y(s).
\end{align}

It is known that \( \left\{ \int_0^t \frac{\sigma_x x(s)}{x(s) + y(s)} dB_x(s) \right\}_{t \geq 0} \) and 
\( \left\{ \int_0^t \frac{\sigma_y y(s)}{x(s) + y(s)} dB_y(s) \right\}_{t \geq 0} \) are martingales. Using Itô's isometry, one obtains

\[
\mathbb{E}\left[ \left( \int_0^t \frac{\sigma_x x(s)}{x(s) + y(s)} dB_x(s) \right)^2 \right] = \mathbb{E}\left( \int_0^t \frac{\sigma_x^2 x^2(s)}{(x(s) + y(s))^2} ds \right) \leq \sigma_x^2 t.
\]

Similarly, \( \mathbb{E}\left[ \left( \int_0^t \frac{\sigma_y y(s)}{x(s) + y(s)} dB_y(s) \right)^2 \right] \leq \sigma_y^2 t \). Then, by the strong law of large numbers for martingales,

\[
\lim_{t \to \infty} \frac{1}{t} \int_0^t \frac{\sigma_x x(s)}{x(s) + y(s)} dB_x(s) = 0 \quad \text{a.s.},
\]
and
\[
\lim_{t \to \infty} \frac{1}{t} \int_0^t \frac{\sigma_y y(s)}{x(s) + y(s)} dB_y(s) = 0 \quad \text{a.s.}
\]

From \eqref{eq3}, dividing by \( t \) and taking the limit as \( t \to \infty \), one obtains:

\[
\lim_{t \to \infty} \frac{1}{t} \ln[x(t) + y(t)] \leq \beta - \min\{\mu_x, \mu_y\}.
\]

Thus, if \( \beta < \min\{\mu_x, \mu_y\} \), then \( \lim_{t \to \infty} \frac{1}{t} \ln[x(t) + y(t)] < 0 \) a.s., and it follows that \( \lim_{t \to \infty} (x(t) + y(t)) = 0 \) a.s. In particular, \( \lim_{t \to \infty} x(t) = \lim_{t \to \infty} y(t) = 0 \) a.s.
\end{proof}

\begin{remark}
    The theorem establishes rather natural conditions. If the fertility rate is below the mortality rate of both juvenile and adult prey, then, regardless of the initial population density, extinction will occur with probability 1.  

\end{remark}

\begin{theorem}[Extinction of the predator]
    Let \( \mathcal{E}_y = \left\{ \lim_{t \to \infty} y(t) = 0 \right\} \). Suppose that \( \mathbb{P}(\mathcal{E}_y) > 0 \). Then, \( \mathbb{P}\left( \lim_{t \to 0} z(t) = 0 \mid \mathcal{E}_y \right) = 1 \). That is, given \( \mathcal{E}_y \), the predator goes extinct with probability \( 1 \).
\end{theorem}
\begin{proof}
Let \( \epsilon \in \left(0, \frac{\frac{1}{2} \sigma_z^2 + \mu_z}{\alpha}\right) \) be arbitrary and for each \( n \in \mathbb{N} \), define \( \mathcal{E}_{\epsilon, n} := \left\{ \sup_{t \geq n} y(t) \leq \epsilon \right\} \). Since \( \mathbb{P}(\mathcal{E}_y) > 0 \) and \( \mathcal{E}_{\epsilon, n} \uparrow \left\{ \limsup_{t \to \infty} y(t) \leq \epsilon \right\} \supseteq \mathcal{E}_y \), there exists a \( N \in \mathbb{N} \) such that:

\[
\mathbb{P}(\mathcal{E}_{\epsilon, n}) > 0, \quad \forall n \geq N.
\]

Conditioning on the event \( \mathcal{E}_{\epsilon, n} \) for \( n \geq N \), it follows that in the stochastic differential equation \eqref{themodel}, the following holds:

\begin{equation}\label{dominacionesto}
    dz \leq (\alpha \epsilon - \mu_z) z \, dt - \sigma_z z \, dB_z(t), \quad \text{for all } t \geq n.
\end{equation}

It is known that if \( \alpha \epsilon - \mu_z < \frac{1}{2} \sigma_z^2 \), the stochastic differential equation

\[
dw = (\alpha \epsilon w - \mu_z w) \, dt - \sigma_z w \, dB(t),
\]

satisfies that for every initial condition \( w_0 \in \mathbb{R}_{> 0} \), \( \lim_{t \to \infty} w(t) = 0 \) whenever \( \alpha \epsilon - \mu_z < \frac{1}{2} \sigma_z^2 \). From \eqref{dominacionesto}, it follows that \( \lim_{t \to \infty} z(t) = 0 \) for every \( z_0 \in \mathbb{R}_{> 0} \). Thus,

\begin{equation}\label{limite}
\mathbb{P}\left(\lim_{t \to \infty} z(t) = 0 \mid \mathcal{E}_{\epsilon, n}\right) = 1.
\end{equation}

Taking the limit as \( n \to \infty \) in \eqref{limite}, one obtains

\begin{equation}\label{limite2}
    \mathbb{P}\left(\lim_{t \to \infty} z(t) = 0 \mid \limsup_{t \to \infty} y(t) \leq \epsilon\right) = 1.
\end{equation}

Finally, taking the limit as \( \epsilon \to 0^+ \), it is concluded that

\[
\mathbb{P}\left(\lim_{t \to \infty} z(t) = 0 \mid \mathcal{E}_y\right) = 1.
\]
\end{proof}

\begin{corollary}
    If \( \beta < \min\{\mu_x, \mu_y\} \), then \( \lim_{t \to \infty} y(t) = 0 \) almost surely, and consequently, \( \lim_{t \to \infty} z(t) = 0 \) almost surely.
\end{corollary}

\subsection{Persistence of the Prey}\label{sec:persistencia}

In this section, the function \( \varphi(x) = \frac{\kappa}{1 + x} \) for \( x \geq 0 \) is considered, where \( \kappa > 0 \). The constant \( \kappa \) can be interpreted as the maturation rate at very low densities. 
In this section, sufficient conditions are provided to ensure that \( \liminf y(t) > 0 \) almost surely in a probability space of the form \( Z_M := \{\sup_{t \geq 0} z(t) \leq M\} \) for some constant \( M > 0 \).

First, some auxiliary results are required. Consider the following system of stochastic differential equations:

\begin{align}\label{aux}
    d\bar{x} &= \left(\beta \bar{y} - \frac{\kappa \bar{x}}{1 + \bar{x}} - \mu_x \bar{x}\right) dt - \sigma_x \bar{x} \, dB_x(t), \\
    d\bar{y} &= \left(\frac{\kappa \bar{x}}{1 + \bar{x}} - (\mu_y + M) \bar{y}\right) dt - \sigma_y \bar{y} \, dB_y(t), \nonumber \\
    d\bar{z} &= -\mu_z \bar{z} \, dt - \sigma_z \bar{z} \, dB_z(t). \nonumber
\end{align}

This system has a generator given by:

\begin{align}
    \bar{\mathcal{L}} := &\left(\beta \bar{y} - \frac{\kappa \bar{x}}{1 + \bar{x}} - \mu_x \bar{x}\right) \frac{\partial}{\partial \bar{x}} + \left(\frac{\kappa \bar{x}}{1 + \bar{x}} - (\mu_y + M) \bar{y}\right) \frac{\partial}{\partial \bar{y}} - \mu_z \bar{z} \frac{\partial}{\partial \bar{z}} \\
    &+ \frac{1}{2} \sigma_x^2 \bar{x}^2 \frac{\partial^2}{\partial \bar{x}^2} + \frac{1}{2} \sigma_y^2 \bar{y}^2 \frac{\partial^2}{\partial \bar{y}^2} + \frac{1}{2} \sigma_z^2 \bar{z}^2 \frac{\partial^2}{\partial \bar{z}^2}. \nonumber
\end{align}

It will be shown that the subsystem of \eqref{aux} given by:

\begin{eqnarray}\label{aux1}
    d\bar{x} &= \left(\beta \bar{y} - \frac{\kappa \bar{x}}{1 + \bar{x}} - \mu_x \bar{x}\right) dt - \sigma_x \bar{x} \, dB_x(t),  \\
    d\bar{y} &= \left(\frac{\kappa \bar{x}}{1 + \bar{x}} - (\mu_y + M) \bar{y}\right) dt - \sigma_y \bar{y} \, dB_y(t), \nonumber
\end{eqnarray}

has a unique stationary distribution with support in \( \mathbb{R}^2_{> 0} \) under certain conditions. To achieve this, it is necessary to show that, under certain conditions, there exists a bounded open set \( \mathcal{U} \subset \bar{U} \subset \mathbb{R}^2_{> 0} \) with a smooth boundary such that the following hold:

\begin{enumerate}
    \item[(i)] There exists a nonnegative function \( V: \mathbb{R}^2_{> 0} \rightarrow \mathbb{R}_{> 0} \) such that \( \mathcal{L}_{sub} V(x, y) < 0 \) for all \( (x, y) \in \mathbb{R}^2_{> 0} \setminus \mathcal{U} \), where \( \mathcal{L}_{sub} \) is the generator of \eqref{aux1}.
    \item[(ii)] There exists \( \kappa > 0 \) such that if \( D \) is the diffusion matrix of \eqref{aux1}, that is, 
    \[
    D = \left(   
    \begin{array}{cc}
       \sigma_x \bar{x} & 0 \\
        0 & \sigma_y \bar{y}
    \end{array}
    \right),
    \]
    then \( (\xi_1, \xi_2) D D^{tr} (\xi_1, \xi_2)^{tr} \geq \kappa (\xi_1^2 + \xi_2^2) \) for all \( (x, y) \in \mathcal{U} \).
\end{enumerate}

The proof of this result can be found in \cite{khasminskii2012stochastic}, p. 99.

\begin{remark}\label{obs}
    The stochastic differential equation \eqref{aux1} satisfies condition (ii), since \[
    (\xi_1, \xi_2) \left(   
    \begin{array}{cc}
       \sigma_x \bar{x} & 0 \\
        0 & \sigma_y \bar{y}
    \end{array}
    \right)
    \left(
    \begin{array}{cc}
       \sigma_x \bar{x} & 0 \\
        0 & \sigma_y \bar{y}
    \end{array}
    \right)
    \left(
    \begin{array}{c}
         \xi_1 \\
         \xi_2
    \end{array}
    \right) = \sigma_x^2 x^2 \xi_1^2 + \sigma_y^2 y^2 \xi_2^2 \geq \kappa (\xi_1^2 + \xi_2^2), \quad \text{for all } (x, y) \in \mathcal{U}.
    \]
\end{remark}

In particular, the existence of a unique stationary distribution implies that for every initial condition \( (x_0, y_0) \in \mathbb{R}^2_{> 0} \), \( \lim_{t \to \infty} (\bar{x}(t), \bar{y}(t)) \overset{d}{=} (X, Y) \), where the random vector \( (X, Y) \) has its support in \( \mathbb{R}^2_{> 0} \). The following conditions are sufficient for the existence of a unique stationary distribution on \( \mathbb{R}^2_{> 0} \) in \eqref{aux1}:

\begin{equation}\label{condiciones}
    \begin{array}{l}
     \frac{1}{2} \sigma_y^2 < \mu_y + M, \\
    \frac{\kappa}{b} + \frac{\mu_x}{b} + \mu_y + M + \frac{\sigma_x^2}{2} + \frac{\sigma_y^2}{2b} < 1.
    \end{array}
\end{equation}

See Theorem \ref{distrestacionaria}.

In event \( Z_M := \{ \sup_{t \geq 0} z(t) \leq M \} \), the stochastic differential equation given by \eqref{ede1} with \( \varphi(x) = \frac{\kappa}{1 + x} \) can be rewritten as:

\begin{align}\label{sistemaequiv}\nonumber
    dx &= \left(\beta y - \frac{\kappa x}{1 + x} - \mu_x x\right) dt - \sigma_x x \, dB_x(t), \\
    dy &= \left(\frac{\kappa x}{1 + x} - \mu_y y - y z \mathbf{1}_{\{z < M\}} \right) dt - \sigma_y y \, dB_y(t), \\
    dz &= \left(\alpha y z - \mu_z z\right) dt - \sigma_z z \, dB_y(t). \nonumber
\end{align}

Using Theorem \ref{comparacion} between \eqref{sistemaequiv} and \eqref{aux} (see the appendix), it follows that if both systems start from the same point \( (x_0, y_0, z_0) \in \mathbb{R}^3_{> 0} \), then

\begin{equation}\label{desigualdad}
    \bar{x}(t) \leq x(t), \quad \bar{y}(t) \leq y(t), \quad \text{and} \quad \bar{z}(t) \leq z(t), \quad \forall t \geq 0 \quad \text{a.s.}
\end{equation}

Therefore, \( \liminf x(t) \geq \lim_{t \to \infty} \bar{x}(t) \overset{d}{=} X \) and \( \liminf y(t) \geq \lim_{t \to \infty} \bar{y}(t) \overset{d}{=} Y \), where the random vector \( (X, Y) \) has the support contained in \( \mathbb{R}^2_{> 0} \). Thus, in space \( Z_M \), conditions \eqref{condiciones} imply persistence of the prey. This result is stated in the following:

\begin{theorem}\label{thm:persistence}
    If the conditions \eqref{condiciones} are satisfied, then 
    \[
    \mathbb{P}(\liminf x(t) > r \mid Z_M) \geq \mathbb{P}(X > r), \quad \forall r > 0,
    \]
    \[
    \mathbb{P}(\liminf y(t) > r \mid Z_M) \geq \mathbb{P}(Y > r), \quad \forall r > 0,
    \]
    where \( (X, Y) \) is a random vector with support in \( \mathbb{R}^2_{> 0} \). In particular,
    \[
    \mathbb{P}(\liminf x(t) > 0 \mid Z_M) = \mathbb{P}(\liminf y(t) > 0 \mid Z_M) = 1.
    \]
\end{theorem}

\begin{remark}
   In the previous theorem, due to the strong Markov property, it is possible to replace $Z_M$ with $Z_M(T):=\sup_{t\geq T}z(t)>M$ for $T>0$, and the conclusion of the theorem remains the same.  
 
\end{remark}

\begin{remark}
   Of course, it would be ideal to present sufficient conditions for the existence of a unique stationary distribution of the model \ref{themodel} over $\mathbb{R}^3_{>0}$. However, this idea could not be fully developed (even in the particular case $\varphi(x)=\frac{\kappa}{1+x}$)) and remains a subject of future work.

\end{remark}

\begin{remark}
  The ideas in the proof of Theorem \ref{thm:persistence} cannot be applied in the case $\varphi(x)=\frac{\kappa }{1+x^2}$, since Theorem \ref{comparacion} cannot be used.

\end{remark}

\section{Simulations }\label{sec:simulaciones}

To illustrate positivity (see Figure \ref{fig:positivity}), the functions $\kappa/(1+x^2)$ and $\kappa/(1+x)$ were chosen as examples. The parameter values were set as follows: $\beta=0.5, \kappa=1, \mu_x=\mu_y=\mu_z=0.1, \alpha=0.2, \sigma_x=\sigma_y=\sigma_z=0.1$.  

It is observed that, despite the relatively large values of $\beta$ and $\kappa$ compared to the mortality rates and the noise intensity terms, the prey population density tends to remain low. The maturation rate at high densities is lower in graph (a), which explains the difference between the graphs given the same parameter sets.

\begin{figure}[h]
    \centering
    % Primera imagen
    \begin{subfigure}{0.45\textwidth}
        \centering
        \includegraphics[width=\linewidth]{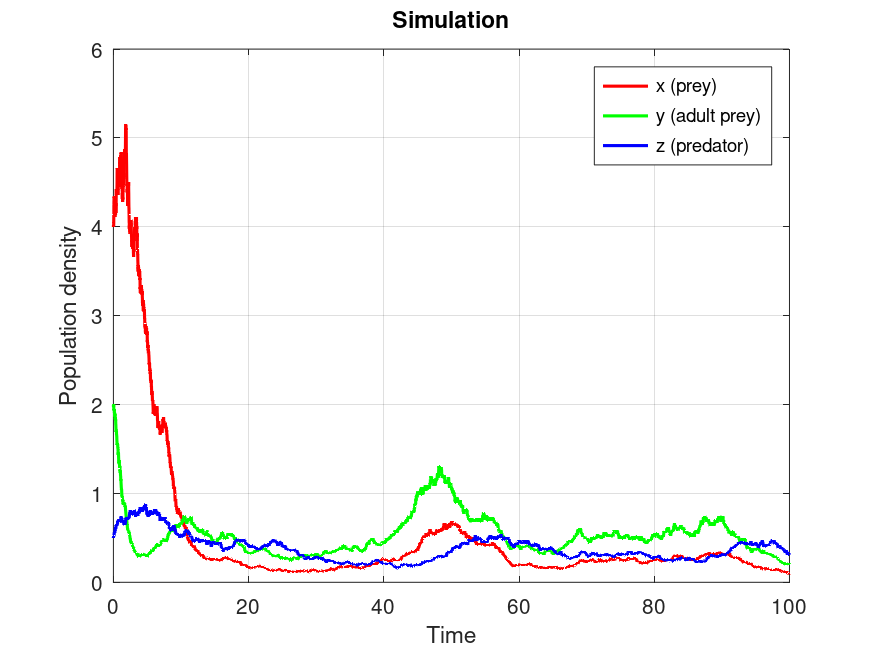}
        \caption{$\varphi(x)=\frac{\kappa }{1+x^2}$}
    \end{subfigure}
    % Segunda imagen
    \begin{subfigure}{0.45\textwidth}
        \centering
        \includegraphics[width=\linewidth]{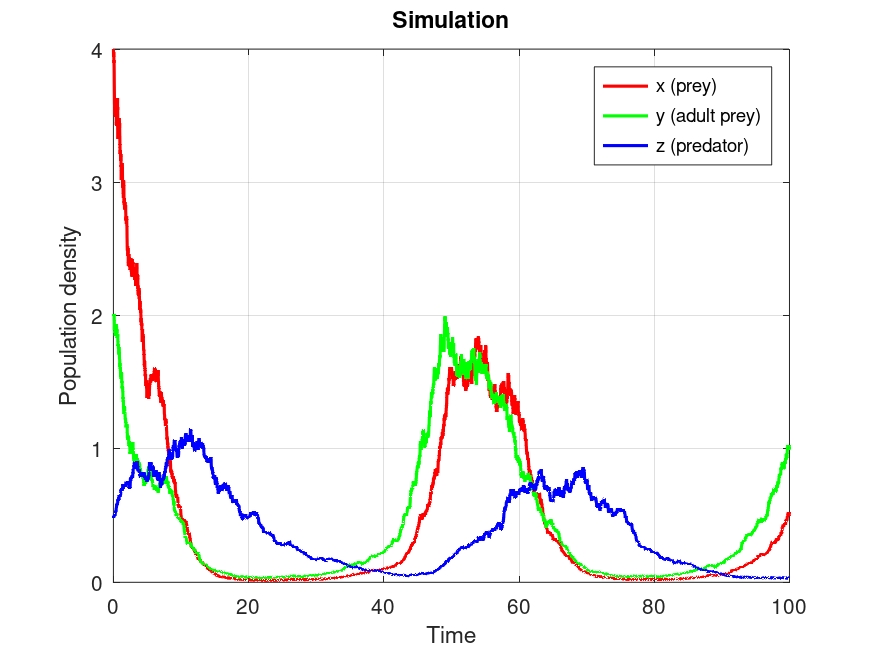}
        \caption{$\varphi(x)=\frac{\kappa}{1+x}$}
    \end{subfigure}
    \caption{Positivity of the solutions.}
    \label{fig:positivity}
\end{figure}

To illustrate the stability in probability of the origin (see Figure \ref{fig:EP}), the following parameter values were used: $\beta=0.05, \kappa=1, \mu_x=\mu_y=\mu_z=0.1, \alpha=0.2, \sigma_x=\sigma_y=\sigma_z=0.1$.  

It is observed that in this case $\mu_y > \beta$ (see Theorem \ref{SP}), which is intuitively clear: if the prey's reproduction rate is lower than its mortality rate, the prey population density will remain low and consequently, the predator will not have enough food available.

\begin{figure}[h]
    \centering
    % Primera imagen
    \begin{subfigure}{0.45\textwidth}
        \centering
        \includegraphics[width=\linewidth]{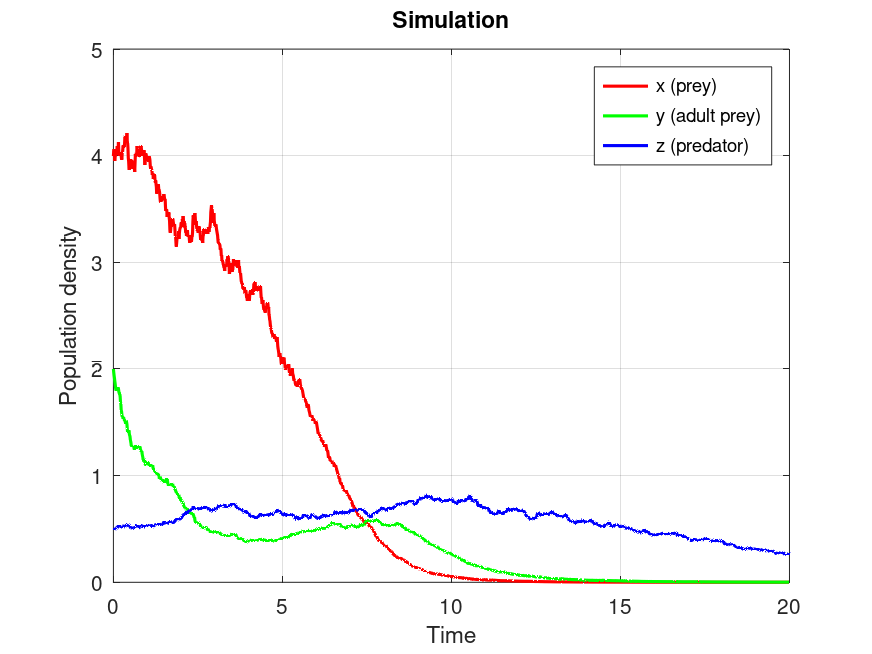}
        \caption{$\varphi(x)=\frac{\kappa }{1+x^2}$}
    \end{subfigure}
    % Segunda imagen
    \begin{subfigure}{0.45\textwidth}
        \centering
        \includegraphics[width=\linewidth]{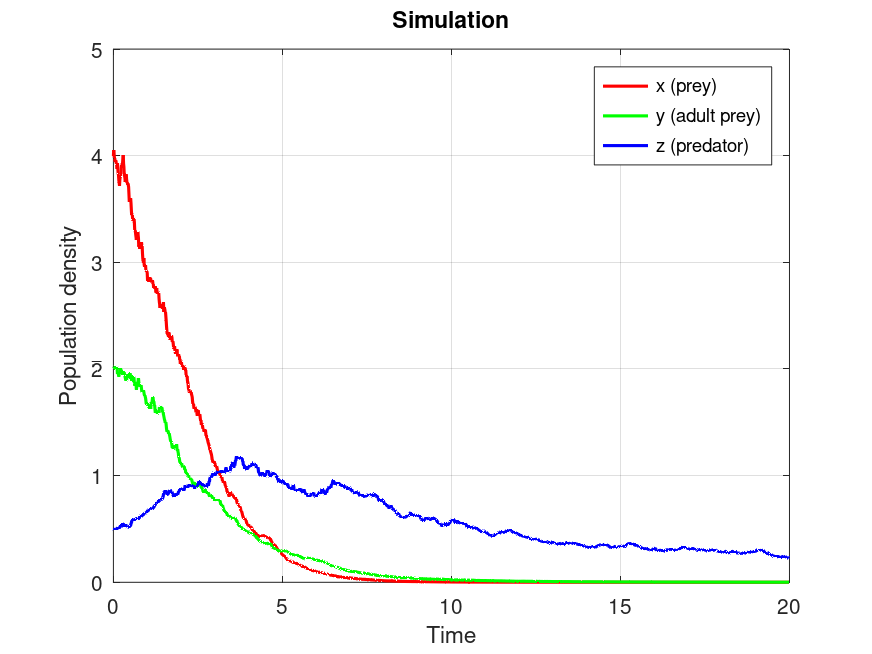}
        \caption{$\varphi(x)=\frac{\kappa}{1+x}$}
    \end{subfigure}
    \caption{Stability in probability of $\mathbf{x}=(0,0,0)$.}
    \label{fig:EP}
\end{figure}

To illustrate the asymptotic stability in probability of the origin (see Figure \ref{fig:EAP}), the following parameter values satisfying the hypotheses of Theorem \ref{sap} are used: \(\beta=0.1\), \(\kappa=1\), \(\mu_x=\mu_y=\mu_z=\sigma_x=\sigma_y=\sigma_z=0.1\), \(\alpha=0.2\). Note that the conditions of the theorem do not involve the parameters \(\alpha\) and \(\kappa\); however, the conditions in Theorem \ref{sap} are not easily interpretable.

\begin{figure}[h]
    \centering
    % Primera imagen
    \begin{subfigure}{0.45\textwidth}
        \centering
        \includegraphics[width=\linewidth]{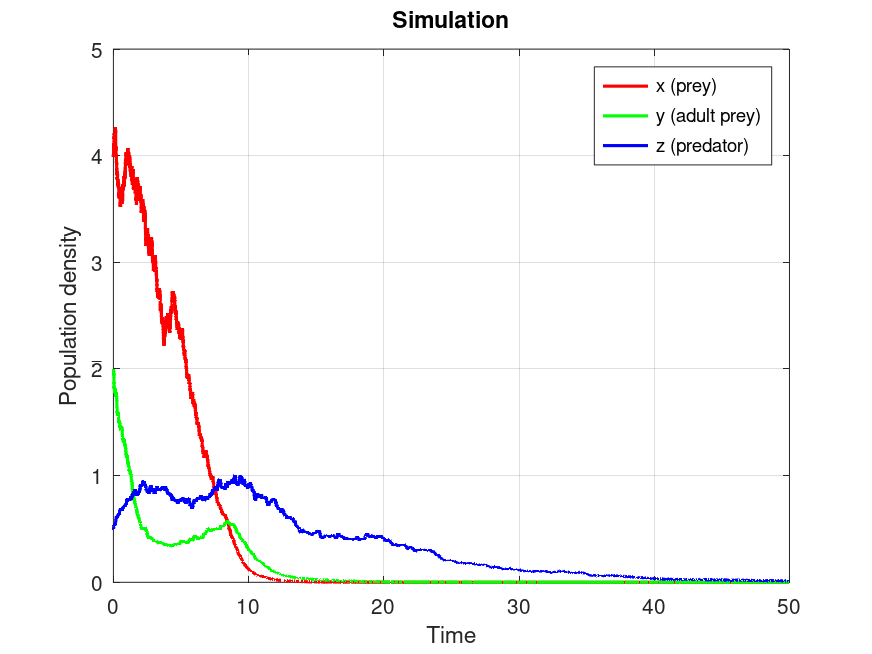}
        \caption{$\varphi(x)=\frac{\kappa }{1+x^2}$}
    \end{subfigure}
    % Segunda imagen
    \begin{subfigure}{0.45\textwidth}
        \centering
        \includegraphics[width=\linewidth]{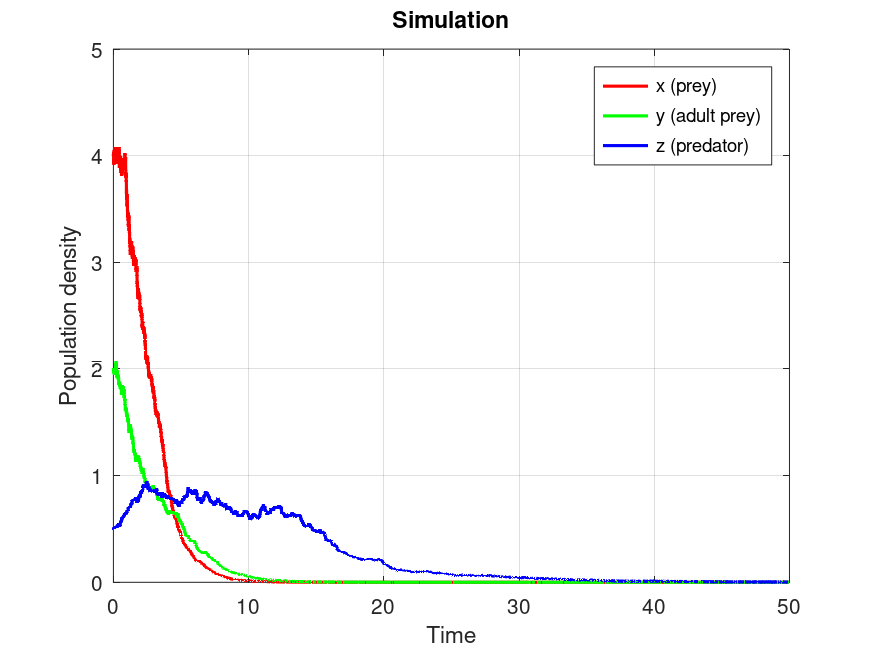}
        \caption{$\varphi(x)=\frac{\kappa}{1+x}$}
    \end{subfigure}
    \caption{Asymptotic stability in probability of $\mathbf{x}=(0,0,0)$.}
    \label{fig:EAP}
\end{figure}

To illustrate the extinction of the prey in the scenario of Theorem \ref{ExtPrey} (see Figure \ref{fig:ExtPrey}), simulations were performed using the following parameters:  
$\beta=0.05, \kappa=1, \mu_x=\mu_y=\mu_z=\sigma_x=\sigma_y=\sigma_z=0.1, \alpha=0.2$.  
In this context, the condition $\beta<\min\{\mu_x,\mu_y\}$ is easily interpretable.

\begin{figure}[h]
    \centering
    % Primera imagen
    \begin{subfigure}{0.45\textwidth}
        \centering
        \includegraphics[width=\linewidth]{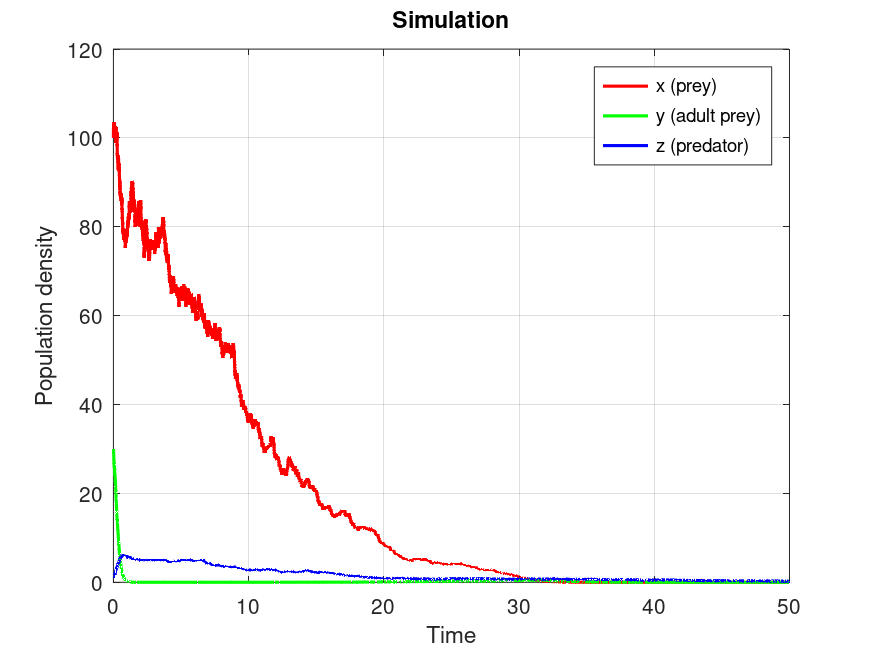}
        \caption{$\varphi(x)=\frac{\kappa }{1+x^2}$}
    \end{subfigure}
    % Segunda imagen
    \begin{subfigure}{0.45\textwidth}
        \centering
        \includegraphics[width=\linewidth]{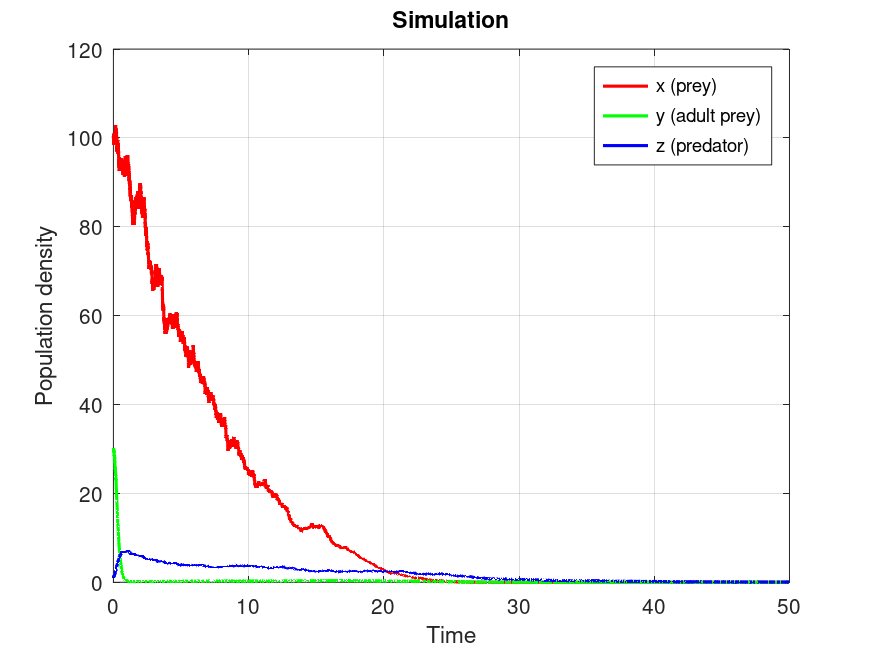}
        \caption{$\varphi(x)=\frac{\kappa}{1+x}$}
    \end{subfigure}
    \caption{Extinction of prey}
    \label{fig:ExtPrey}
\end{figure}

In the following scenario, it is found that despite the predator remaining at low population density thresholds, the prey also remains at very low thresholds. This could indicate the non-persistence of the prey over time. The parameters used were as follows:  
\[
(x_0,y_0,z_0)=(100, 30, 5), \quad \beta=0.1, \quad \kappa=1, \quad \mu_x=\mu_y=\sigma_x=\sigma_y=\sigma_z=0.1, \quad \mu_z=0.6, \quad \alpha=0.7.
\]
Note that for the case \(\varphi(x)=\frac{\kappa}{1+x}\), the conditions of Theorem \ref{thm:persistence} are not satisfied.
  
ANow, in the case $\varphi(x)=\frac{\kappa}{1+x}$ with the following parameters, the conditions of Theorem \ref{thm:persistence} are guaranteed, with $M=0.5$ and $(x_0,y_0,z_0)=(4, 1.5, 0.4)$, $\beta =2$, $\kappa=0.3$, $\mu_x=\mu_y=\sigma_x=\sigma_y=\sigma_z=0.1$, $\mu_z=0.6$, and $\alpha=0.7$. See Figure \ref{fig:persistence}, where the black dashed line corresponds to the value $M=0.5$. This figure suggests the existence of a stationary distribution for model \ref{themodel} under these parameters.

\begin{figure}[h]
    \centering
    \includegraphics[width=0.8\linewidth]{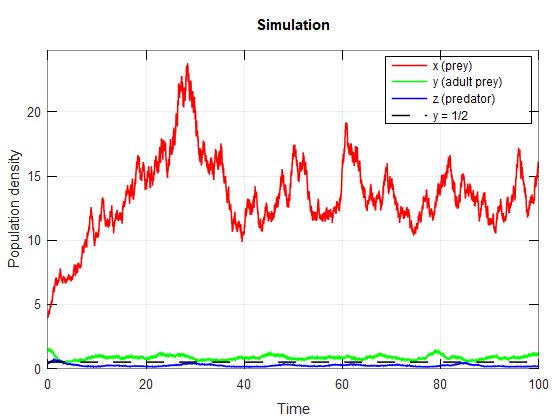} % Ajusta el ancho al 80% del ancho de línea
    \caption{$\varphi(x)=\frac{\kappa }{1+x}$}
    \label{fig:persistence}
\end{figure}

With this in mind, the histogram of the values of $(x(t),y(t))$ for $t=1000$ is constructed. Indeed, the numerical results suggest the presence of a stationary distribution (see Figure \ref{fig:stationary}).

\begin{figure}[h]
    \centering
    \includegraphics[width=0.8\linewidth]{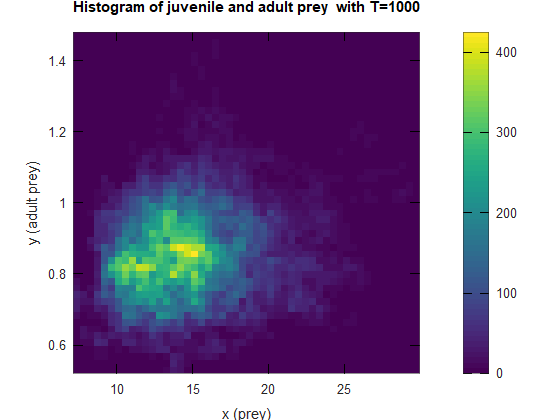} % Ajusta el ancho al 80% del ancho de línea
    \caption{$\varphi(x)=\frac{\kappa }{1+x}$}
    \label{fig:stationary}
\end{figure}

\section{Appendix}\label{sec:apendice}

The comparison theorem is widely used in the theory of stochastic differential equations (see \cite{Geiss1994}). Below, the preliminaries necessary to state the Comparison Theorem are presented. Let \( (\Omega, \mathcal{F}, \mathbb{P}) \) be a complete probability space, and let \( \{\mathcal{F}_t\}_{t \geq 0} \) be a right-continuous filtration in \( \mathcal{F} \) that contains all \( \mathbb{P} \)-null sets of \( \mathcal{F} \). Consider the following systems of stochastic differential equations defined in the same space:

\begin{equation}\label{X}
    dX_j(t)=a_j(X(t),t)dt+\sum_{k=1}^r\sigma_{jk}(X(t),t)dB_k(t),
\end{equation}
y

\begin{equation}\label{Y}
    dY_j(t)=b_j(Y(t),t)dt+\sum_{k=1}^r\sigma_{jk}(Y(t),t)dB_k(t),
\end{equation}

for \( j = 1, \ldots, d \), \( t \geq 0 \), and \( X_j(0) \leq Y_j(0) \) for \( j = 1, \ldots, d \) almost surely, where \( B = \{(B_1(t), \ldots, B_r(t))\}_{t \geq 0} \) is an \( r \)-dimensional Brownian motion with respect to \( \{\mathcal{F}_t\}_{t \geq 0} \). Let \( X(0) \) and \( Y(0) \) be \( \mathcal{F}_0 \)-measurable random variables. The coefficients \( a = (a_1, \ldots, a_d) \) and \( \sigma = (\sigma_{jk})_{1 \leq j \leq d, 1 \leq k \leq r} \) are assumed to be continuous functions defined in \( \mathbb{R}^d \times \mathbb{R}_{\geq 0} \). Let \( X \) and \( Y \) be the sample paths of local solutions to \eqref{X} and \eqref{Y}, respectively, on \( (\Omega, \mathcal{F}, \mathbb{P}) \) with respect to \( \{\mathcal{F}_t\}_{t \geq 0} \) and \( (B, \{\mathcal{F}_t\}_{t \geq 0}) \). Let \( \tau_X \) and \( \tau_Y \), respectively, be the explosion times.

\begin{theorem}\label{comparacion}
    Assume that the following conditions are satisfied:

\begin{enumerate}
    \item For any \( t \geq 0 \) and \( j = 1, \ldots, d \), it holds that \( a_j(\mathbf{x}, t) \leq b_j(\mathbf{y}, t) \) for any \( x_j = y_j \) and \( x_l \leq y_l \) for all \( l \neq j \).
    \item There exists a strictly increasing function \( \rho: \mathbb{R}_{\geq 0} \rightarrow \mathbb{R}_{\geq 0} \) with \( \rho(0) = 0 \) and \( \int_{\mathbb{R}_{\geq 0}} \rho^{-2}(u) \, du = \infty \) such that for each \( j = 1, \ldots, d \),
    \[
    \sum_{k=1}^r |\sigma_{jk}(\mathbf{x}, t) - \sigma_{jk}(\mathbf{y}, t)| \leq \rho(|x_j - y_j|) \quad \text{for all } t \geq 0, \, \mathbf{x}, \mathbf{y} \in \mathbb{R}^d.
    \]
\end{enumerate}

Then, \( X(0) \leq Y(0) \) almost surely implies \( \mathbb{P}(X(t) \leq Y(t), \, t \in [0, \tau_X \wedge \tau_Y]) = 1 \).
\end{theorem}

\begin{proof}
    See Theorem 1.1 and Theorem 1.2 in \cite{Geiss1994}.
\end{proof}
\begin{theorem}\label{distrestacionaria}
    The stochastic differential equation \eqref{aux1} has a unique stationary distribution on \( \mathbb{R}^2_{> 0} \) if the following conditions are satisfied:

\begin{enumerate}
    \item \( \frac{1}{2} \sigma_y^2 < \mu_y + M \),
    \item \( \frac{\kappa}{b} + \frac{\mu_x}{b} + (\mu_y + M) + \frac{\sigma_x^2}{2} + \frac{\sigma_y^2}{2\beta} < 1 \).
\end{enumerate}
\end{theorem}
\begin{proof}

The theorem in \cite{khasminskii2012stochastic}, p. 99, is used. That is, if \( \mathcal{L}_{sub} \) is the generator of \eqref{aux1}, then:

\begin{eqnarray*}
    \mathcal{L}_{sub} & := & \left(\beta y - \frac{\kappa x}{1 + x} - \mu_x x\right) \frac{\partial}{\partial x} + \left(\frac{\kappa x}{1 + x} - \gamma y \right) \frac{\partial}{\partial y} \\
    & & + \sigma_x x^2 \frac{\partial^2}{\partial x^2} + \sigma_y y^2 \frac{\partial^2}{\partial y^2},
\end{eqnarray*}

where \( \gamma = \mu_y + M \). Then, by Remark \eqref{obs}, it is sufficient to show that there exists a bounded open set \( \mathcal{U} \subset \bar{\mathcal{U}} \subset \mathbb{R}^2_{> 0} \) and a function \( V: \mathbb{R}^2_{> 0} \rightarrow \mathbb{R}_{> 0} \) of class \( C^2 \) such that:
\[
\mathcal{L}_{sub} V(x, y) \leq -c, \quad \forall (x, y) \in \mathbb{R}^2_{> 0} \setminus \mathcal{U},
\]
for some positive constant \( c \). Now, let \( V(x, y) = \frac{1}{\beta}(x - \ln x) + y - \ln y + \frac{1}{2} y^2 \). In this way, it follows that:
\begin{eqnarray*}
    \mathcal{L}_{sub}V(x,y) & = &\frac{1}{\beta} \left(\beta y -\frac{\kappa x}{1+x}-\mu_x x\right)\left(1-\frac{1}{x} \right)+\left( \frac{\kappa x}{1+x}-\gamma y  \right)\left( 1-\frac{1}{y}+y \right)\\
    & & +\frac{\sigma_x^2}{2\beta}+\frac{\sigma_y^2}{2}\\
    &\leq & (1-\gamma)y - \frac{\kappa x}{\beta(1+x)}- \frac{\mu_xx}{\beta}-\frac{y}{x}+\frac{\kappa}{\beta}+\frac{\mu_x}{\beta}+\kappa- \frac{\kappa x}{y(1+x)}+\gamma\\
    & & +\frac{\kappa xy}{1+x}+\left(\frac{\sigma_y^2}{2}-\gamma \right)y^2+\frac{\sigma_x^2}{\beta}+\frac{\sigma_y^2}{2}, \quad \forall (x,y)\in \mathbb{R}^2_{>0}.
\end{eqnarray*} 

Let \( Q = (\epsilon_1, 1/\epsilon_2) \times (\epsilon_3, 1/\epsilon_4) \), where \( \epsilon_i > 0 \) for \( i = 1, 2, 3, 4 \) are constants to be determined. It will be shown that \( \mathcal{L}_{sub} V(x, y) \leq -c \) for all \( (x, y) \in \mathbb{R}^2_{> 0} \setminus Q \), for some positive constant \( c \). This is sufficient, since any open set contained in \( Q \) with a smooth boundary satisfies the required conditions.

The set \( \mathbb{R}^2_{> 0} \setminus Q \) is divided as follows:

\[
\mathbb{R}^2_{> 0} \setminus Q = \bigcup_{i=1}^4 \Lambda_i,
\]

where

\begin{itemize}
    \item \( \Lambda_1 = \{ (x, y) \in \mathbb{R}^2_{> 0} : 0 < x \leq \epsilon_1 \} \),
    \item \( \Lambda_2 = \{ (x, y) \in \mathbb{R}^2_{> 0} : x > 1/\epsilon_2 \} \),
    \item \( \Lambda_3 = \{ (x, y) \in \mathbb{R}^2_{> 0} : 0 < y \leq \epsilon_3 \} \),
    \item \( \Lambda_4 = \{ (x, y) \in \mathbb{R}^2_{> 0} : y \geq 1/\epsilon_4 \} \).
\end{itemize}

\textbf{Case 1:} If \( (x, y) \in \Lambda_1 \), consider the following function:

\[
g(x, y) = (1 - \gamma) y - \frac{y}{x} - \frac{\kappa x}{y(1 + x)} + \frac{\kappa x y}{1 + x} + \left( \frac{\sigma_y^2}{2} - \gamma \right) y^2.
\]

Then,

\[
\mathcal{L}_{sub} V(x, y) \leq g(x, y) + \frac{\kappa}{\beta} + \frac{\mu_x}{\beta} + \kappa + \gamma + \frac{\sigma_y^2}{2} + \frac{\sigma_x^2}{2\beta}, \quad \text{on } \Lambda_1.
\]

After a series of straightforward calculations, we obtain:

$$
\sup_{(x,y)\in \Lambda_1}g(x,y)=-1-\kappa, 
$$

By the conditions in the theorem, \( \gamma < \frac{\sigma_y^2}{2} \) and \( \frac{\kappa}{b} + \frac{\mu_x}{b} + \mu_y + \frac{\sigma_x^2}{2} + \frac{\sigma_y^2}{2\beta} < 1 \). Therefore, there exists a positive constant \( c_1 \) such that for any \( \epsilon_1 > 0 \) (taking \( \epsilon_1 = 1 \) for simplicity), it holds that:

\begin{equation}\label{Lambda1}
    \mathcal{L}_{sub} V(x, y) \leq -c_1, \quad \text{on } \Lambda_1.
\end{equation}

\textbf{Case 2:} If \( (x, y) \in \Lambda_2 \), by condition 1 in the hypotheses of the theorem, we have:

\begin{eqnarray*}
    \mathcal{L}_{sub} V(x, y) &\leq& - \frac{\kappa x}{\beta(1 + x)} - \frac{\mu_x x}{\beta} - \frac{y}{x} - \frac{\kappa x}{y(1 + x)} + \frac{\kappa}{\beta} + \frac{\mu_x}{\beta} + \kappa + \gamma + \frac{\sigma_y^2}{2} + \frac{\sigma_x^2}{2\beta} \\
    && + y \left( 1 - \gamma + \frac{\kappa x}{1 + x} \right) + \left( \frac{\sigma_y^2}{2} - \gamma \right) y^2.
\end{eqnarray*}

Due to condition 1, it follows that \( \lim_{x \to +\infty} \mathcal{L}_{sub} V(x, y) = -\infty \). Therefore, there exist positive constants \( \epsilon_2 \in (0, 1) \) and \( c_2 \) such that:

\begin{equation}\label{Lambda2}
    \mathcal{L}_{sub} V(x, y) \leq -c_2, \quad \text{on } \Lambda_2.
\end{equation}

\textbf{Case 3:} If \( (x, y) \in \Lambda_3 \). From Cases 1 and 2, it follows that we only need to consider the case where \( x \in (1, 1/\epsilon_2) \). Define \( h: \mathbb{R}_{> 0} \rightarrow \mathbb{R} \) by:

\begin{equation*}
    h(y) = (1 - \gamma) y - \frac{y}{1/\epsilon_2} - \frac{\kappa}{y(1 + (1/\epsilon_2))} + \frac{\kappa y}{(1 + (1/\epsilon_2))} + \left( \frac{\sigma_y^2}{2} - \gamma \right) y^2.
\end{equation*}

Then,

\[
h(y) \leq g(x, y), \quad \forall (x, y) \in \Lambda_3 \cap (1, 1/\epsilon_2) \times \mathbb{R}_{> 0}.
\]

Since \( \lim_{y \to 0^+} h(y) = -\infty \), there exist \( \epsilon_3 > 0 \) and a positive constant \( c_3 \) such that:

\begin{equation}\label{Lambda3}
    \mathcal{L}_{sub} V(x, y) \leq -c_3, \quad \text{on } \Lambda_3.
\end{equation}

\textbf{Case 4:} If \( (x, y) \in \Lambda_4 \). Considering only the case where \( x \in (1, 1/\epsilon_2) \) and by condition 1 in the hypotheses, we have:

\begin{eqnarray*}
    \mathcal{L}_{sub} V(x, y) &\leq& -\frac{\kappa x}{\beta(1 + x)} - \frac{\mu_x x}{\beta} - \frac{y}{x} - \frac{\kappa x}{y(1 + x)} + \frac{\kappa}{\beta} + \frac{\mu_x}{\beta} + \kappa + \gamma + \frac{\sigma_y^2}{2} + \frac{\sigma_x^2}{2\beta} \\
    && + y \left( 1 - \gamma + \frac{\kappa x}{1 + x} \right) + \left( \frac{\sigma_y^2}{2} - \gamma \right) y^2.
\end{eqnarray*}

Since the function \( \frac{\kappa x}{1 + x} \) is bounded, it follows that \( \lim_{y \to +\infty} \mathcal{L}_{sub} V(x, y) = -\infty \). Therefore, there exist \( \epsilon_4 \in (0, 1) \) and \( c_4 > 0 \) such that:

\begin{equation}\label{Lambda4}
    \mathcal{L}_{sub} V(x, y) \leq -c_4, \quad \text{on } \Lambda_4.
\end{equation}

From \eqref{Lambda1}, \eqref{Lambda2}, \eqref{Lambda3}, and \eqref{Lambda4}, it follows that if \( c = \min_{1 \leq i \leq 4} c_i \), then:

\[
\mathcal{L}_{sub} V(x, y) \leq -c, \quad \forall (x, y) \in \mathbb{R}^2_{> 0} \setminus Q.
\]

\end{proof}

%Bibliography
\bibliographystyle{plain}  % O cualquier otro estilo válido
\bibliography{references}

\end{document}